\documentclass[journal,10pt]{IEEEtran}
\usepackage{epsfig}
\usepackage{booktabs}
\usepackage{epstopdf}
\usepackage{amssymb}
\usepackage{mathtools}

\usepackage{algorithmic}

\usepackage{array}

\usepackage{mdwmath}
\usepackage{mdwtab}
\usepackage{bm}
\usepackage{mathrsfs}
\usepackage{epsfig}
\usepackage{subfigure}
\usepackage{algorithmic}
\usepackage{amsmath}
\usepackage{mathrsfs}
\usepackage{color}
\usepackage{amsmath}
\usepackage{amsmath}
\usepackage{url}
\newtheorem{theorem}{Theorem}
\newtheorem{lemma}{Lemma}
\newenvironment{proof}{{\noindent\it Proof}\quad}{\hfill $\square$\par}
\newtheorem{prop}{Proposition}[section]
\newtheorem{coro}{Corollary}[section]
\newtheorem{remark}{Remark}

\begin{document}

\title{A Simple Evaluation for the Secrecy Outage Probability Over Generalized-$K$ Fading Channels}

\author{Hui Zhao,~\IEEEmembership{Student Member,~IEEE,} Yuanwei Liu,~\IEEEmembership{Senior Member,~IEEE,} \\ Ahmed Sultan-Salem,~\IEEEmembership{Member,~IEEE,}  and Mohamed-Slim Alouini,~\IEEEmembership{Fellow,~IEEE} \vspace{-1cm}

\thanks{Manuscript received June 5, 2019; accepted June 25, 2019. The work of H. Zhao was done while he was studying at KAUST. The associate editor coordinating the review of this paper and approving it for publication was Y. Deng. (\emph{Corresponding author: Yuanwei Liu.})}

\thanks{H. Zhao was with the Computer, Electrical, and Mathematical Science and Engineering Division, King Abdullah University of Science and Technology (KAUST), Thuwal 23955-6900, Saudi Arabia, and he is now with the Communication Systems Department, EURECOM, Sophia Antipolis 06410, France (email: hui.zhao@kaust.edu.sa).}

\thanks{Y. Liu is with the School of Electronic Engineering and Computer Science, Queen Mary University of London, London E1 4NS, U.K (email: yuanwei.liu@qmul.ac.uk).}

\thanks{A. Sultan-Salem and M.-S. Alouini are with the Computer, Electrical, and Mathematical Science and Engineering Division, King Abdullah University of Science and Technology, Thuwal 23955-6900, Saudi Arabia (email: ahmed.salem@kaust.edu.sa; slim.alouini@kaust.edu.sa).}

\thanks{Color versions of one or more of the figures in this paper are available online at http://ieeexplore.ieee.org.}

\thanks{Digital Object Identifier 10.1109/LCOMM.2019.2926360}
}

\markboth{IEEE Communications Letters,~Vol.~XX, No.~XX, XXX~2019}
{}

\maketitle

\begin{abstract}
A simple approximation for the secrecy outage probability (SOP) over generalized-$K$ fading channels is developed. This approximation becomes  tighter as the average signal-to-noise ratio (SNR) of the wiretap channel decreases.
Based on this simple expression, we also analyze the asymptotic SOP in the high SNR region of the main channel. Besides simplifying the SOP expression significantly, this asymptotic SOP expression reveals the secrecy diversity order in a general case. Numerical results demonstrate the high accuracy of our proposed approximation results.
\end{abstract}


\IEEEpeerreviewmaketitle

\vspace{-0.5cm}
\section{Introduction}
The Generalized-$K$ (GK) fading model was proposed in \cite{Bithas} to approximate the composite fading channel (small-scale fading plus large-scale fading). Some important metrics, such as outage probability and ergodic capacity, were analyzed in \cite{Bithas}-\cite{Lateef}. However, as there is a modified Bessel function of the second kind in the probability density function (PDF) of the GK fading model, it is usually difficult to derive the closed-form expressions for some important metrics in more complicated models. To address this issue, the authors in \cite{Atapattu} proposed a mixture Gamma distribution method to approximate the PDF of GK fading with a high accuracy when the number of summation terms becomes large.

Physical layer security has the potential to address the emerging security issues in modern wireless networks \cite{Bloch}-\cite{Zhang_TVT}. The security under many fading channel models, such as Rayleigh, Nakagami-$m$, Lognormal, and Nakagami-$m$/Gamma models, has been analyzed \cite{Juan_TCOM}-\cite{Peppas}.  The authors in \cite{Lei_CL} firstly investigated the typical Wyner's three-node model of \cite{Bloch} in physical layer security over GK fading channels, and derived the closed-form expression for the secrecy outage probability (SOP) by using the mixture Gamma distribution approximation introduced in \cite{Atapattu}. This mixture Gamma approximation was also adopted in many works about the secure analysis over GK fading channels, such as \cite{Lei_FITEE}-\cite{Zhao_CL}.  Further, based on the mixture Gamma distribution model in \cite{Atapattu}, the authors in \cite{Kang_WCL} developed a general method to analyze the SOP by using the Fox's $H$-function.
Although the authors in \cite{Lei_IET,Lei_TVT} used the exact PDF of the GK fading model to derive the SOP expression, the SOP expression was derived based on the assumption of a large average signal-to-noise ratio (SNR) of the wiretap channel, which means that the derived results will deviate the exact results significantly in the low SNR region of the wiretap channel.
The secure approximation analysis in \cite{Pan_TVT} for Log-normal and Log-normal-Rayleigh composite fading channels based on the work of \cite{Holtzman} shows a strong robustness in the low variance region of the wiretap channel, where the difference calculation takes place of the integration calculation, resulting in a much faster computation for the SOP.
Although the asymptotic SOP (ASOP) valid in the high SNR region of the main channel, showing the secrecy diversity order and array gain, was analyzed in \cite{Lei_FITEE,Zhao_CL}, the authors only considered a special parameter setting of GK fading, which cannot be used in the general GK fading case.

In this letter, we consider the approximation method proposed by \cite{Pan_TVT,Holtzman} to derive a simple and robust closed-form expression for the SOP in the typical Wyner's three-node model over GK fading channels. When the average SNR of the main channel is sufficiently large,  the asymptotic analysis in the general case is also presented to get the secrecy diversity order and array gain, which is useful and important for the secure system design.

\section{System Model}
In the typical Wyner's three-node model, there is a source ($S$) transmitting confidential message to a destination ($D$), while an eavesdropper ($E$) wants to overhear the information from $S$ to $D$. We assume that all links undergo independent GK fading. To be realistic, we consider a silent eavesdropping scenario where $S$ does not know the channel state information of the $S-E$ link. In this case, perfect security cannot be guaranteed, because $S$ has to adopt a constant rate of confidential message ($R_s$).
The SOP is the probability that the secrecy capacity ($C_s$) is less then $R_s$ \cite{Bloch}, where $C_s$ is defined as
$
{C_s} =\max\{ {\log _2}\left( {1 + {\gamma _d}} \right) - {\log _2}\left( {1 + {\gamma _e}} \right), 0\},
$
where $\gamma_d$ and $\gamma_e$ are the instantaneous SNRs at $D$ and $E$, respectively, and $\max\{a,b\}=a$ if $a \ge b$, or otherwise $\max\{a,b\}=b$.
Thus, we can write the SOP as
\begin{align}\label{SOP}
&P_{\rm{sop}} = \Pr \left\{ {{{\log }_2}\left( {1 + {\gamma _d}} \right) - {{\log }_2}\left( {1 + {\gamma _e}} \right) \le {R_s}} \right\} \notag\\
&= \int\limits_0^\infty  {{F_{{\gamma _d}}}\left( {\lambda  - 1 + \lambda x} \right)} {f_{{\gamma _e}}}\left( x \right)dx
= {\mathbb{E}_{{\gamma _e}}}\left\{ {{F_{{\gamma _d}}}\left( {\lambda  - 1 + \lambda {\gamma _e}} \right)} \right\},
\end{align}
where $\lambda=2^{R_s}$, $f_{\gamma_t} (\cdot)$ and $F_{\gamma_t} (\cdot)$ are the PDF and cumulative density function (CDF) of $\gamma_t \in \{\gamma_d, \gamma_e\}$, respectively, and $\mathbb{E}\{\cdot\}$ denotes the expectation operator.

In this letter, we adopt a general approximation to calculate the SOP in \eqref{SOP}.  Following \cite{Holtzman}, we approximate $P(X)$, a real-valued function of a random variable $X$ with mean $\mu$ and variance $\sigma^2$, using the $N$-th degree Taylor polynomial, i.e.,
\begin{align}\label{P}
P\left( X \right) \simeq \sum\nolimits_{n = 0}^N {\frac{{{{\left( {X - \mu } \right)}^n}}}{{n!}}} {P^{\left( n \right)}}\left( \mu  \right),
\end{align}
where $P^{(n)}(X)$ denotes the $n$-th derivative of $P(X)$ with respect to $X$. Taking expectation of both sides in \eqref{P}, we have
\begin{align}
\mathbb{E}\left\{ {P\left( X \right)} \right\} \simeq \sum\nolimits_{n = 0}^N {\frac{{{P^{\left( n \right)}}\left( \mu  \right)}}{{n!}}} \mathbb{E}\left\{ {{{\left( {X - \mu } \right)}^n}} \right\}.
\end{align}
If $P(x)$ has  non-zero higher order derivatives, we have an approximation error,  which unfortunately may not  decrease with increasing the summation terms  in general cases. \cite{Holtzman} proposed that $N=2$ is a good approximation with the compromise between complexity and accuracy, if the variance of $X$ is not too large. When $N=2$, the approximate expectation becomes
\begin{align}\label{P_appro}
\mathbb{E}\left\{ {P\left( X \right)} \right\} \simeq P\left( \mu  \right) + \frac{{{\sigma ^2}{P^{\left( 2 \right)}}\left( \mu  \right)}}{2}.
\end{align}
By using the approximation in \eqref{P_appro}, the SOP in \eqref{SOP} can be approximated as
\begin{align}\label{SOP_appro}
P_{\rm{sop}} \simeq P\left( {{{\overline \gamma  }_e}} \right) + \frac{{\sigma _e^2{P^{\left( 2 \right)}}\left( {{{\overline \gamma  }_e}} \right)}}{2},
\end{align}
where $P\left( {{\overline\gamma _e}} \right) = {F_{{\gamma _d}}}\left( {\lambda  - 1 + \lambda {\overline\gamma _e}} \right)$, $\overline \gamma_t$ and $\sigma_t^2$ are the mean and variance of $\gamma_t \in \{\gamma_d,\gamma_e\}$, respectively.
\section{SOP Over GK Fading Channels}
In this section, the closed-form expression for the SOP over GK fading channels will be given, along with two special cases of the GK model, i.e., Rayleigh and Nakagami-$m$ fading channels.
\begin{theorem}
\emph{In the typical Wyner's three-node model, the closed-form expression for the SOP over GK fading channels can be approximated by}
\begin{align}\label{SOP_final}
&{P_{{\rm{sop}}}} = \frac{{G_{1,3}^{2,1}\left( {\frac{{{k_d}{m_d}\left( {\lambda  - 1 + \lambda {{\overline \gamma  }_e}} \right)}}{{{{\overline \gamma  }_d}}}\left| {_{{k_d}{m_d},0}^1} \right.} \right)}}{{\Gamma \left( {{k_d}} \right)\Gamma \left( {{m_d}} \right)}} + \overline \gamma  _e^2{\lambda ^2} \notag\\
&\frac{{\left( {\frac{{\left( {{k_e} + 1} \right)\left( {{m_e} + 1} \right)}}{{{k_e}{m_e}}} - 1} \right)G_{2,4}^{2,2}\left( {\frac{{{k_d}{m_d}\left( {\lambda  - 1 + \lambda {{\overline \gamma  }_e}} \right)}}{{{{\overline \gamma  }_d}}}\left| {_{{k_d},{m_d},0,2}^{0,1}} \right.} \right)}}{{2\Gamma \left( {{k_d}} \right)\Gamma \left( {{m_d}} \right){{\left( {\lambda  - 1 + \lambda {{\overline \gamma  }_e}} \right)}^2}}},
\end{align}
\emph{where $m_t>0$ and $k_t>0$ ($t \in \{d,e\}$) are the small-scale and large-scale fading parameters of the GK fading model, respectively, and $G^{\cdot,\cdot}_{\cdot,\cdot}(\cdot)$ denotes the Meijer's G-function \cite{Gradshteyn}. This approximation is tight for small $\overline \gamma_e$. When $\overline \gamma_e$ is too large, (15) in \cite{Lei_IET} derived based on $\overline \gamma_e \gg 0$ can be used to derive the SOP. Therefore, an approximate and simple SOP expression with a high accuracy can be always derived for any $\overline \gamma_e$.}
\end{theorem}
\begin{proof}
The CDF of $\gamma_t$ ($\gamma_t \in \{\gamma_d,\gamma_e\}$)  over GK fading channels can be written in terms of the Meijer's G-function as \cite{Lei_TVT}
\begin{align}\label{gamma_d_CDF}
{F_{{\gamma _t}}}\left( {{\gamma _t}} \right) = \frac{{G_{1,3}^{2,1}\left( {\frac{{{k_t}{m_t}{\gamma _t}}}{{{{\overline \gamma  }_t}}}\left| {_{{k_t},{m_t},0}^1} \right.} \right)}}{{\Gamma \left( {{k_t}} \right)\Gamma \left( {{m_t}} \right)}}.
\end{align}
The corresponding $P(\overline\gamma_e)$ in \eqref{SOP_appro} becomes
\begin{align}
P\left( {{\overline \gamma _e}} \right) = \frac{{G_{1,3}^{2,1}\left( {\frac{{{k_d}{m_d}\left( {\lambda  - 1 + \lambda {\overline \gamma _e}} \right)}}{{{{\overline \gamma  }_d}}}\left| {_{{k_d},{m_d},0}^1} \right.} \right)}}{{\Gamma \left( {{k_d}} \right)\Gamma \left( {{m_d}} \right)}}.
\end{align}
To derive the $n$-th ($n=0,1,2,\cdots$) order derivative of $P(\overline \gamma_e)$, we can employ the $n$-th order derivative property of Meijer's G-function, given by \cite{Wolfram}
\begin{align}
{z^n}\frac{{{\partial ^n}}}{{\partial {z^n}}}G_{p,q}^{m,n}\left( {z\left| {_{{{\bf{b}}_q}}^{{{\bf{a}}_p}}} \right.} \right) = G_{p + 1,q + 1}^{m,n + 1}\left( {z\left| {_{{{\bf{b}}_q},n}^{0,{{\bf{a}}_p}}} \right.} \right),
\end{align}
where ${\bf a}_p$ and ${\bf b}_q$ represent the parameter vectors of Meijer's G-function, respectively.
By using this derivative identity, the $n$-th derivative of $P(\overline \gamma_e)$ can be easily derived as
\begin{align}\label{P_derivative}
&P^{(n)}(\overline\gamma_e) = \frac{\lambda ^nG_{2,4}^{2,2}\left( {\frac{{{k_d}{m_d}\left( {\lambda  - 1 + \lambda {\overline\gamma _e}} \right)}}{{{{\overline \gamma  }_d}}}\left| {_{{k_d},{m_d},0,n}^{0,1}} \right.} \right)}{{\Gamma \left( {{k_d}} \right)\Gamma \left( {{m_d}} \right){{\left( {\lambda  - 1 + \lambda {\gamma _e}} \right)}^n}}}.
\end{align}

The $n$-th moment function of $\gamma_t$ over GK fading channels is given by (5) in \cite{Bithas}
\begin{align}
\mathbb{E}\left\{ {\gamma _t^n} \right\} = \frac{{\Gamma \left( {{k_t} + n} \right)\Gamma \left( {{m_t} + n} \right)}}{{\Gamma \left( {{k_t}} \right)\Gamma \left( {{m_t}} \right)}}{\left( {\frac{{{{\overline \gamma  }_t}}}{{{k_t}{m_t}}}} \right)^n}.
\end{align}
The variance of $\gamma_e$ can be easily derived by using this moment function, given by
\begin{align}\label{moment_E}
\sigma _e^2 = \mathbb{E}\left\{ {\gamma _e^2} \right\} - {\mathbb{E}^2}\left\{ {{\gamma _e}} \right\}
 = \left[ {\frac{{\left( {{k_e} + 1} \right)\left( {{m_e} + 1} \right)}}{{{k_e}{m_e}}} - 1} \right]\overline \gamma  _e^2,
\end{align}
where a large $\overline \gamma_e$ means a large variance of $\gamma_e$.

Substituting the derived derivative of $P(\overline \gamma_e)$ and variance of $\gamma_e$ into \eqref{SOP_appro} yields \eqref{SOP_final}.
\end{proof}

\begin{remark}
The SOP expression shown in \emph{Theorem 1} provides a simple approximation compared to the one in \cite{Kang_WCL} where the SOP expression involves the infinite univariate Meijer's G-function summation. Unlike the SOP expression in \cite{Lei_CL} valid only for integer $m_d$ and $m_e$, our derived SOP expression can be used for arbitrary positive $m_d$ and $m_e$.
\end{remark}

\begin{coro}
For $k_d=k_e \to \infty$ and $m_d=m_e=1$, the GK fading model is reduced to the Rayleigh fading model.   The approximate SOP over Rayleigh fading channels is
\begin{align}\label{SOP_Rayleigh}
P_{\rm sop} = 1 - \left[ {1 - \frac{1}{2}{{\left( {\frac{{{{\overline \gamma  }_e}\lambda }}{{{{\overline \gamma  }_d}}}} \right)}^2}} \right]\exp \left( { - \frac{{\lambda  - 1 + \lambda {{\overline \gamma  }_e}}}{{{{\overline \gamma  }_d}}}} \right).
\end{align}
\end{coro}

\begin{coro}
When $k_d=k_e \to \infty$, i.e., no shadowing, the GK model becomes the Nakagami-$m$ fading model. The SOP over Nakagami-$m$ fading channels is
\begin{align}\label{SOP_Naka}
{P_{{\rm{sop}}}} =& \frac{{\Upsilon \left( {{m_d},\frac{{{m_d}\left( {\lambda  - 1 + \lambda {{\overline \gamma  }_e}} \right)}}{{{{\overline \gamma }_d}}}} \right)}}{{\Gamma \left( {{m_d}} \right)}} + \overline \gamma  _e^2{\lambda ^2} \notag\\
&\frac{{\left( {\frac{{{m_e} + 1}}{{{m_e}}} - 1} \right)G_{2,3}^{1,2}\left( {\frac{{{m_d}\left( {\lambda  - 1 + \lambda {{\overline \gamma  }_e}} \right)}}{{{{\overline \gamma }_d}}}\left| {_{{m_d},0,2}^{0,1}} \right.} \right)}}{{2\Gamma \left( {{m_d}} \right){{\left( {\lambda  - 1 + \lambda {\overline\gamma _e}} \right)}^2}}},
\end{align}
where $\Upsilon(\cdot,\cdot)$ denotes the lower incomplete Gamma function \cite{Gradshteyn}.
\end{coro}

\begin{proof}
The proof of \emph{Corollary 3.1} (or \emph{Corollary 3.2}) is straightforward by substituting $m_d=m_e=1$ and $k_d=k_e \to \infty$ (or $k_d=k_e \to \infty$) into \eqref{SOP_final}.
\end{proof}

In the Nakagami-$m$ fading case, our SOP approximation, shown in \emph{Corollary 3.2}, is much concise compared to (34) in \cite{Juan_TCOM}, where the SOP expression is composed by multiple summation terms. In addition, the SOP approximation in \cite{Juan_TCOM} is valid only for integer $m$.

\section{Asymptotic Analysis}
In this section, the asymptotic expression for SOP will be presented over GK fading channels, according to two cases, i.e., $k_d \neq m_d$ and $k_d=m_d$, where $k_d>0$ and $m_d>0$.
\begin{lemma}
\emph{When $\overline \gamma_d \to \infty$ and $\overline \gamma_e$ is finite, the ASOP for $k_d \neq m_d$ over GK fading channels is given by}
\begin{align}\label{ASOP_mnk}
P_{{\rm{sop}}}^\infty  =& \overline \gamma  _d^{ - v}\frac{{\Gamma \left( {\left| {{k_d} - {m_d}} \right|} \right){{\left( {{k_d}{m_d}} \right)}^v}{{\left( {\lambda  - 1 + \lambda {{\overline \gamma  }_e}} \right)}^v}}}{{\Gamma \left( {{k_d}} \right)\Gamma \left( {{m_d}} \right)}} \notag\\
&\left( {\frac{1}{v} + \frac{{\left( {\frac{{\left( {{k_e} + 1} \right)\left( {{m_e} + 1} \right)}}{{{k_e}{m_e}}} - 1} \right)\left( {v - 1} \right)\overline \gamma  _e^2{\lambda ^2}}}{{{{2\left( {\lambda  - 1 + \lambda {{\overline \gamma  }_e}} \right)}^2}}}} \right),
\end{align}
\emph{where $v=\min\{k_d,m_d\}$, i.e., $v$ is the minimum between $k_d$ and $m_d$. It is obvious that the secrecy diversity order is $\min\{k_d,m_d\}$.}
\end{lemma}

\begin{proof}
When $\overline \gamma_d \to \infty$ and $k_d \neq m_d$, the CDF of $\gamma_d$ can be approximated as \cite{Lateef,Zhao_CL}
\begin{align}
F_{{\gamma _d}}^\infty \left( x \right) = \frac{{\Gamma \left( {\left| {{k_d} - {m_d}} \right|} \right){{\left( {{k_d}{m_d}x} \right)}^v}}}{{\Gamma \left( {{k_d}} \right)\Gamma \left( {{m_d}} \right)v}}\overline \gamma  _d^{ - v},
\end{align}
The asymptotic result for the Meijer's G-function in \eqref{P_derivative} for $n=2$ can be approximated by using the series expansion up to the first non-zero order term,
\begin{align}\label{MeijerG_appro}
&G_{2,4}^{2,2}\left( {\frac{{{k_d}{m_d}\left( {\lambda  - 1 + \lambda {{\overline \gamma  }_e}} \right)}}{{{{\overline \gamma  }_d}}}\left| {_{{k_d},{m_d},0,2}^{0,1}} \right.} \right) \notag\\
& \simeq \left( {v - 1} \right){\left( {{k_d}{m_d}} \right)^v}{\left( {\lambda  - 1 + \lambda {{\overline \gamma  }_e}} \right)^v}\Gamma \left( {\left| {{k_d} - {m_d}} \right|} \right)\overline \gamma  _d^{ - v},
\end{align}
which can be easily derived by following (14)-(16) in \cite{Zhao_CL2} by rewriting the Meijer's G-function into the integral form and computing the residue at the double pole, i.e., the leading term in high SNRs.
We can derive \eqref{ASOP_mnk} by using the asymptotic CDF of $\gamma_d$ for $k_d \neq m_d$ and the asymptotic expression for $P^{(2)}(\cdot)$.
\end{proof}

\begin{prop}
The asymptotic CDF of $\gamma_d$ for $k_d=m_d$ and $\overline \gamma_d \to \infty$ is
\begin{align}
F_{{\gamma _d}}^\infty \left( {{\gamma _d}} \right) = \frac{{\psi \left( {{m_d} + 1} \right) + 2\psi \left( 1 \right) - \psi \left( {{m_d}} \right) - \ln \left( {\frac{{m_d^2{\gamma _d}}}{{{{\overline \gamma  }_d}}}} \right)}}{{m_d^{ - 2{m_d} + 1}{\Gamma ^2}\left( {{m_d}} \right)\gamma _d^{ - {m_d}}\overline \gamma  _d^{{m_d}}}},
\end{align}
where $\psi(\cdot)$ denotes the digamma function \cite{Gradshteyn}.
\end{prop}

\begin{proof}
The proof is similar to \eqref{MeijerG_appro} by referring to (14)-(16) in \cite{Zhao_CL2}.
\end{proof}

\begin{remark}
By using the derived asymptotic CDF of $\gamma_d$ for $k_d=m_d$ in \emph{Proposition 4.1} and the diversity order definition, the diversity order can be derived by
\begin{align}
 &- \mathop {\lim }\limits_{{{\overline \gamma  }_d} \to \infty } \frac{{\ln F_{{\gamma _d}}^\infty \left( {{\gamma _d}} \right)}}{{\ln {{\overline \gamma  }_d}}} =m_d,
\end{align}
which shows that the diversity order for $k_d=m_d$ is still $\min\{k_d,m_d\}$.
\end{remark}

\begin{lemma}
\emph{The closed-form expression for the  ASOP for $k_d=m_d$ over GK fading channels is }
\begin{align}
\resizebox{.9\hsize}{!}{$
P_{{\rm{sop}}}^\infty  = \overline \gamma _d^{ - {m_d}}\left( {\frac{{\psi \left( {{m_d} + 1} \right) + 2\psi \left( 1 \right) - \psi \left( {{m_d}} \right) - \ln \left( {\frac{{m_d^2\left( {\lambda  - 1 + \lambda {{\overline \gamma }_e}} \right)}}{{{{\overline\gamma }_d}}}} \right)}}{{m_d^{ - 2{m_d} + 1}{\Gamma ^2}\left( {{m_d}} \right)(\lambda-1+\lambda \overline\gamma_e)^{ - {m_d}}}}} \right. $} \notag
\end{align}
\vspace{-0.5cm}
\begin{align}\label{ASOP_mk}
\resizebox{.9\hsize}{!}{$
+ \left. {\frac{{\left( {\frac{{\left( {{k_e} + 1} \right)\left( {{m_e} + 1} \right)}}{{{k_e}{m_e}}} - 1} \right)\left( {\psi \left( {{m_d} - 1} \right) + 2\psi \left( 1 \right) - \psi \left( {{m_d}} \right) - \ln \frac{{m_d^2\left( {\lambda  - 1 + \lambda {{\overline \gamma }_e}} \right)}}{{{{\overline \gamma }_d}}}} \right)}}{{2m_d^{ - 2{m_d}}{\lambda ^{ - 2}}\overline \gamma  _e^{ - 2}\Gamma \left( {{m_d}} \right)\Gamma \left( {{m_d} - 1} \right){{\left( {\lambda  - 1 + \lambda {{\overline \gamma }_e}} \right)}^{2 - {m_d}}}}}} \right).
$}
\end{align}
\end{lemma}

\begin{proof}
In the $k_d=m_d$ case, the second derivative of $P(\overline \gamma_e)$ can be written as
\begin{align}\label{P_asy}
&{P^{\left( 2 \right)}}\left( {{{\overline \gamma  }_e}} \right) = \frac{{{\lambda ^2}G_{2,4}^{2,2}\left( {\frac{{m_d^2\left( {\lambda  - 1 + \lambda {{\overline \gamma  }_e}} \right)}}{{{{\overline \gamma  }_d}}}\left| {_{{m_d},{m_d},0,2}^{0,1}} \right.} \right)}}{{{\Gamma ^2}\left( {{m_d}} \right){{\left( {\lambda  - 1 + \lambda {{\overline \gamma  }_e}} \right)}^2}}} \notag\\
&\mathop  \simeq \limits^{\left( a \right)}  \frac{{\psi \left( {{m_d} - 1} \right) + 2\psi \left( 1 \right) - \psi \left( {{m_d}} \right) - \ln \frac{{m_d^2\left( {\lambda  - 1 + \lambda {{\overline \gamma  }_e}} \right)}}{{{{\overline \gamma  }_d}}}}}{{m_d^{ - 2{m_d}}{\lambda ^{ - 2}}\Gamma \left( {{m_d}} \right)\Gamma \left( {{m_d} - 1} \right){{\left( {\lambda  - 1 + \lambda {{\overline \gamma  }_e}} \right)}^{2 - {m_d}}}\overline \gamma  _d^{{m_d}}}},
\end{align}
where $(a)$ follows the series expansion of the Meijer's-G function at $\overline \gamma_d \to \infty$. Specifically, this approximation can be easily obtained by referring to (14)-(16) in \cite{Zhao_CL2}.

Substituting the asymptotic expression for $P^{(2)}(\cdot)$ in \eqref{P_asy} and the CDF of $\gamma_d$ derived in \emph{Proposition 4.1} into \eqref{SOP_appro} yields \eqref{ASOP_mk}.
\end{proof}

\emph{Lemma 2} shows that the ASOP is not a linear function with respect to $\overline \gamma_d$ in dB, because of $\ln \overline \gamma_d$ in \eqref{ASOP_mk}. However, the secrecy diversity order is  $\min\{k_d,m_d\}$, and the slope of ASOP changes very slowly in high SNRs of the main channel with respect to $\ln \overline \gamma_d$.
To best of authors' knowledge, the ASOP expression for $m_d=k_d$ is presented for the first time.

As $\Gamma(m_d-1)$ and $\psi(m_d-1)$ approach to infinity for $m_d=1$ in \eqref{ASOP_mk}, we also give the specific expression for $k_d=m_d=1$, i.e., Rayleigh-Gamma ($K$-distribution) composite fading of the $S-D$ link,
\begin{align}
\resizebox{.9\hsize}{!}{$
P_{{\rm{sop}}}^\infty  = \overline \gamma  _d^{ - 1}\left( {\frac{{\psi \left( 1 \right) + \psi \left( 2 \right) - \ln \left( {\frac{{\lambda  - 1 + \lambda {{\overline \gamma  }_e}}}{{{{\overline \gamma  }_d}}}} \right)}}{{{{\left( {\lambda  - 1 + \lambda {{\overline \gamma  }_e}} \right)}^{ - 1}}}} - \frac{{\overline \gamma  _e^2{\lambda ^2}\left( {\frac{{\left( {{k_e} + 1} \right)\left( {{m_e} + 1} \right)}}{{{k_e}{m_e}}} - 1} \right)}}{{2\left( {\lambda  - 1 + \lambda {{\overline \gamma  }_e}} \right)}}} \right).
$}
\end{align}

\section{Numerical Results}
In this section, we use the Monte-Carlo simulation to validate the high accuracy of our proposed SOP and ASOP expressions, i.e., \eqref{SOP_final}, \eqref{ASOP_mnk} and \eqref{ASOP_mk}.

In Figs. 1-3, the improving trend of SOP is obvious with increasing $\overline \gamma_d$ (or $m_d$ or $k_d$), because of the improved average main channel state (or the increasing multi-path or lighter shadowing). As shown in Fig. 1, apart from a growing SOP for a large $\overline \gamma_e$ due to the improved wiretap channel, we can easily see that the deviation between the simulation and our proposed SOP approximation results (i.e., \eqref{SOP_final})  converges with decreasing $\overline \gamma_e$, i.e., smaller variance of $\gamma_e$, in the medium $\overline \gamma_d$ region. Generally, the approximate SOP results match the simulation results very well even for a large $\overline \gamma_e$ (for $\overline \gamma_e=15$ dB, the corresponding $\sigma_e^2$ is 1100). Figs. 1-3 also shows that when $\overline \gamma_d$ is sufficiently large, the difference among three SOP results (i.e., simulation, approximation and asymptotic results) almost vanishes, which is valid for any value of $\overline \gamma_e$.

\begin{figure}[!htb]
\setlength{\abovecaptionskip}{0pt}
\setlength{\belowcaptionskip}{10pt}
\centering
\includegraphics[width= 3 in]{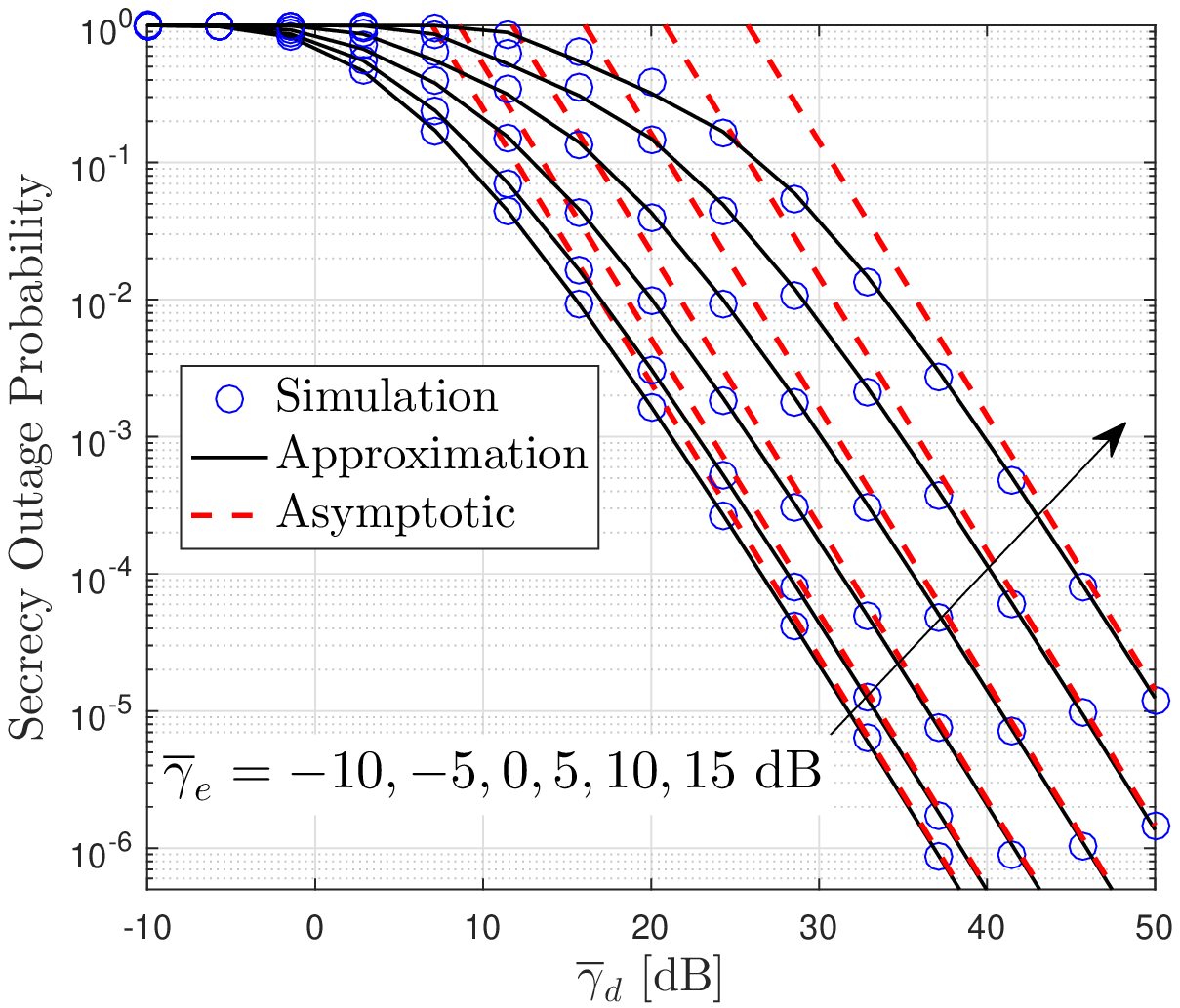}
\caption{$P_{\rm sop}$ versus $\overline \gamma_d$ for $m_d=m_e=2.5$, $k_d=k_e=2$, and $R_s=1$.}
\label{P_D_muH}
\end{figure}

From Fig. 2, for $k_d=1.5$, the slope of ASOP changes for $m_d=0.5, 1$, and becomes constant for $m_d=2, 2.5$, which shows that the secrecy diversity order is $\min\{k_d,m_d\}$. Fig. 3 plots the SOP versus $\overline \gamma_d$ for different $k_d=m_d$ values. The ASOP is not a linear function with respect to $\overline \gamma_d$ in dB, while the slope of ASOP changes very slowly in the high $\overline \gamma_d$ region.

\begin{figure}[!htb]
\setlength{\abovecaptionskip}{0pt}
\setlength{\belowcaptionskip}{10pt}
\centering
\includegraphics[width= 3 in]{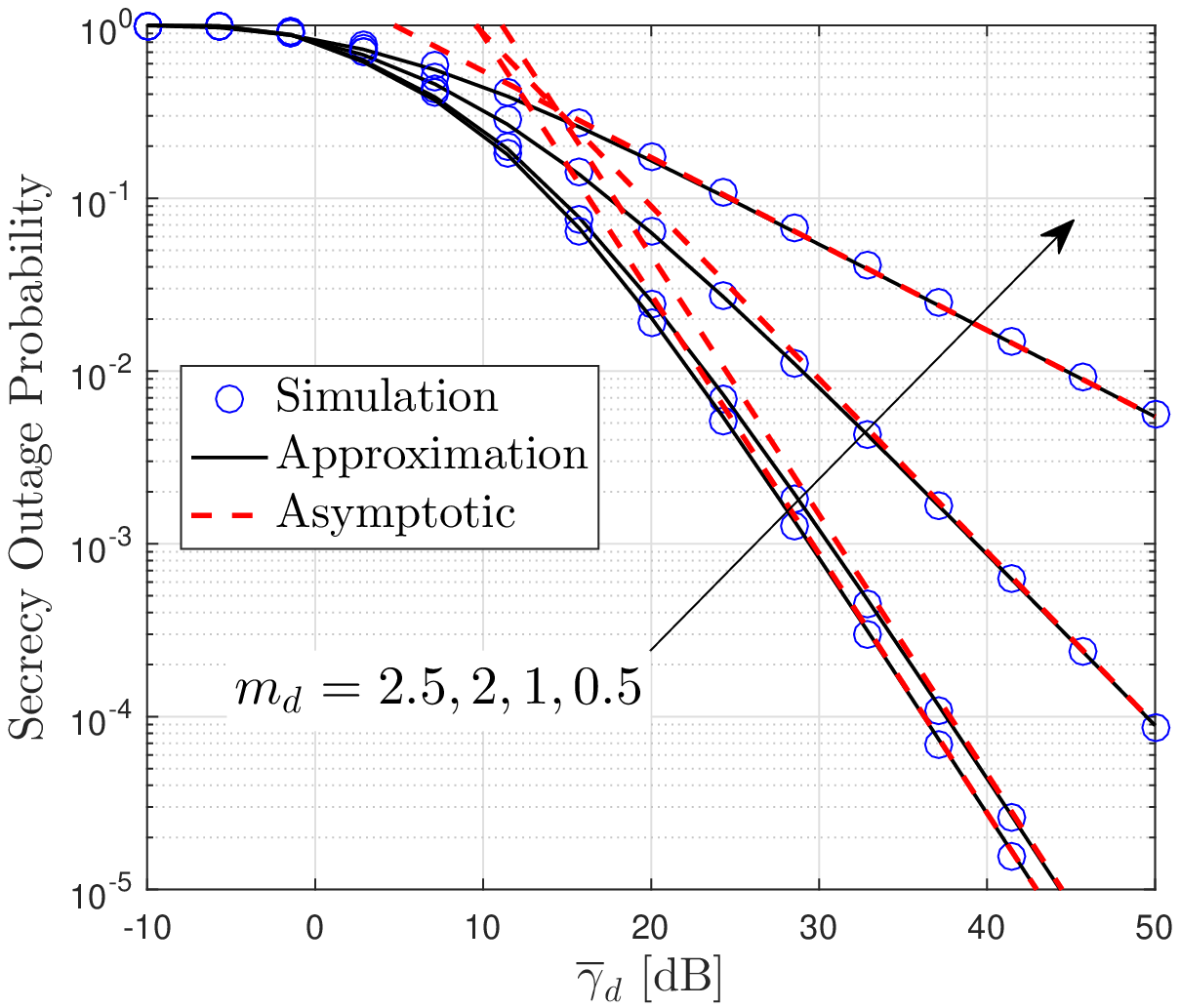}
\caption{$P_{\rm sop}$ versus $\overline \gamma_d$ for $k_d=m_e=k_e=1.5$, $\overline \gamma_e=0$ dB, and $R_s=1$.}\vspace{-0.7cm}
\label{P_D_muH}
\end{figure}
\begin{figure}[!htb]
\setlength{\abovecaptionskip}{0pt}
\setlength{\belowcaptionskip}{10pt}
\centering
\includegraphics[width= 3 in]{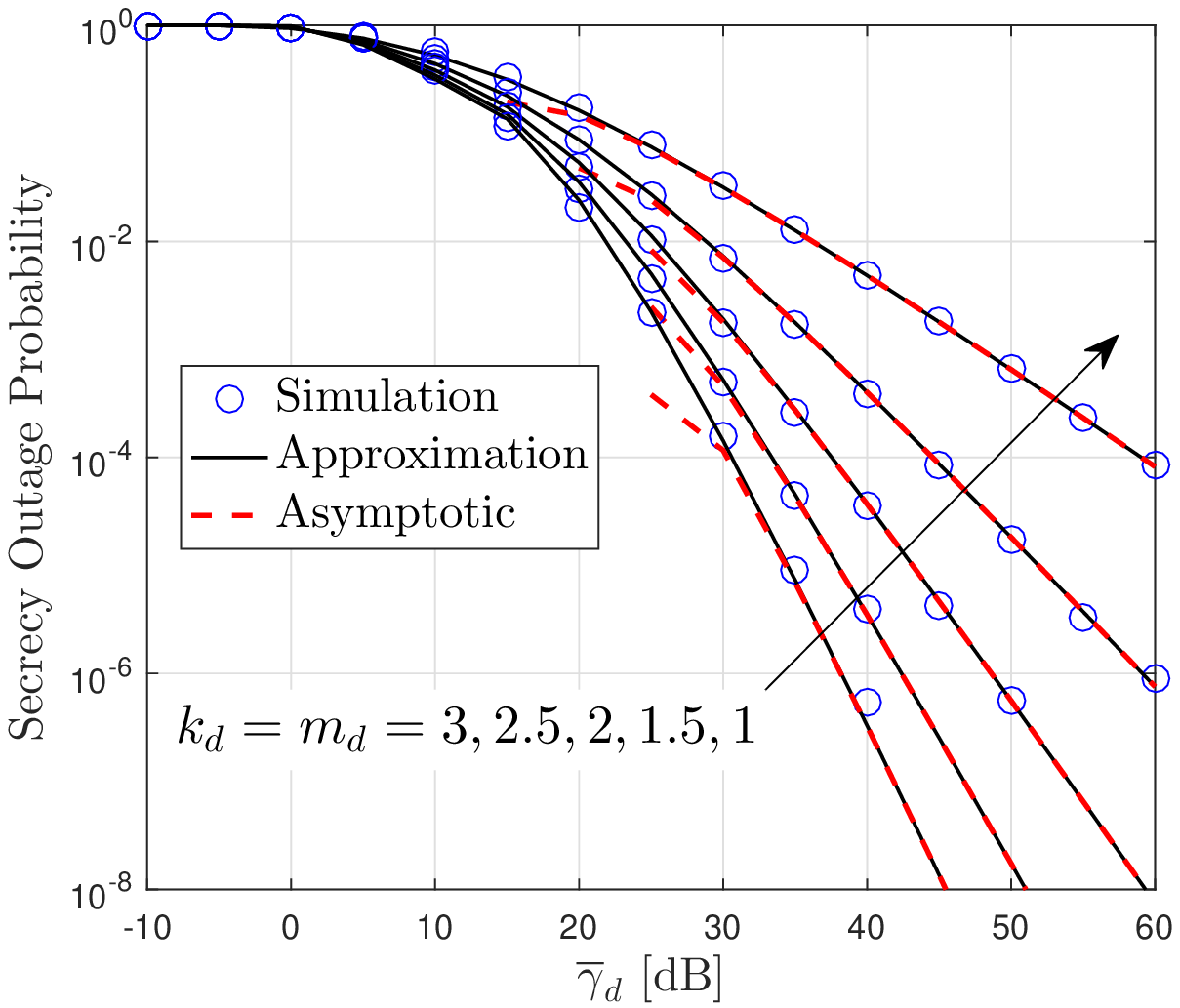}
\caption{$P_{\rm sop}$ versus $\overline \gamma_d$ for $m_e=k_e=2$, $\overline \gamma_e=5$ dB, and $R_s=1$.}
\label{P_D_muH}
\end{figure}

\section{Conclusion}
In this letter, a simple SOP expression was derived with a high accuracy for a small average SNR of the wiretap channel. Although the matching becomes worse for a large average SNR of the wiretap channel, the approximate SOP converges to the exact SOP with increasing the average SNR of the main channel. To obtain the secrecy diversity order, we also derived the asymptotic expression for the SOP valid in the high SNR region of the main channel. The ASOP expression also shows that the ASOP for $k_d=m_d$ is not a linear function with respect to $\overline \gamma_d$ in dB, despite the slowly changing  slope in high SNRs.



\end{document}